\documentclass[12pt,oneside,english]{amsart}
\usepackage[latin9]{inputenc}
\usepackage{geometry}
\usepackage{amsthm}
\usepackage{amssymb}
\usepackage{setspace}
\makeatletter

\newcommand{\noun}[1]{\textsc{#1}}

\numberwithin{equation}{section}
\numberwithin{figure}{section}
\numberwithin{table}{section}
\theoremstyle{plain}
\newtheorem{thm}{\protect\theoremname}[section]
\theoremstyle{plain}
\newtheorem{lem}[thm]{\protect\lemmaname}
\theoremstyle{remark}
\newtheorem{rem}[thm]{\protect\remarkname}
\theoremstyle{definition}
\newtheorem{defn}[thm]{\protect\definitionname}
\theoremstyle{plain}
\newtheorem{prop}[thm]{\protect\propositionname}
\ifx\proof\undefined
\newenvironment{proof}[1][\protect\proofname]{\par
\normalfont\topsep6\p@\@plus6\p@\relax
\trivlist
\itemindent\parindent
\item[\hskip\labelsep\scshape #1]\ignorespaces
}{%
\endtrivlist\@endpefalse
}
\providecommand{\proofname}{Proof}
\fi
\theoremstyle{plain}
\newtheorem{cor}[thm]{\protect\corollaryname}
\newtheorem{notation}[thm]{Notation}

\usepackage{babel}
\providecommand{\corollaryname}{Corollary}
\providecommand{\definitionname}{Definition}
\providecommand{\lemmaname}{Lemma}
\providecommand{\propositionname}{Proposition}
\providecommand{\remarkname}{Remark}
\providecommand{\theoremname}{Theorem}

\begin{document}

\title{potential examples for non-additivity of the minimal output entropy}

\maketitle
\author{M. Al Nuwairan}

\address{King Faisal university}

\email{msalnuwairan@kfu.edu.sa}

\makeatother

\makeatother

\begin{abstract}
In this paper, we study the minimal output entropy of EPOSIC channels.
We determine the cases where their minimal output entropy is zero,
and obtain some partial results on the fulfillment of their entanglement
breaking property. Our results show that these channels provide potential
examples for non-additivity of the minimal output entropy.
\end{abstract}

\section{introduction}

The carrier of the states (information) from one part to another in
quantum systems is known as a quantum channel. Ideally, a channel
carries a state form one system to another without losing information.
However, the existence of noise in all information processing systems
affects the channel's performance in any transmission of such information.
One of the important open questions is that of determining the capability
of a channel to transmit classical information, which is known as
the classical capacity of the channel. In their attempts to increase
the capacity of quantum channels, scientists studied whether or not,
running two channels in parallel will increase the total classical
capacity of two channels. Failing to do so, the capacity is called
additive. According to \cite[Prop 8.2]{key-7}, and P. Shor in \cite{key-12},
the additivity of another quantity known as the minimal output entropy
(MOE) of the channel implies the additivity of the classical capacity.
Much research effort was directed to prove the additivity of the minimal
output entropy. It has been proved for some special classes of quantum
channels such as channels with zero minimal output entropy, tensoring
the identity with any channel \cite{key-2}, and the entanglement
breaking channels \cite{key-11}. However, an outstanding paper in
2008 by Hastings \cite{key-4} disproved this conjecture. By giving
a randomized construction of channels that violates the additivity
of the minimal output entropy, he was able to show that there exists
an example of a channel $\Phi$ such that $S_{min}(\Phi\otimes\overline{\Phi})\neq S_{min}(\Phi)+S_{min}(\overline{\Phi})$.
Since then, the efforts redirected towards constructing an explicit
example for the non-additivity of the minimal output entropy. In this
paper, we provide a potential solution of this problem.

In \cite{key-1}, EPOSIC channels were introduced. They are non random
quantum channels that form the extreme points of all $SU(2)$-irreducibly
covariant channels. Here, we show that large classes of these channels
have nonzero minimal output entropy, and they are not entanglement
breaking. Hence, they form potential examples for violating the additivity
conjecture. The next section contains definition of EPOSIC channels,
we precisely determine the cases where the EPOSIC channels have zero
minimal output entropy, and compute the minimal output entropy for
some of EPOSIC channels. In section III, we obtain partial results
on the fulfillment of the entanglement breaking property of EPOSIC
channels. All vector spaces considered in this paper are finite dimensional.

Our main results are:
\begin{itemize}
\item For $m,n,h\in\mathbb{N}$ such $h\leq\min\{m,n\}$, the channel $\Phi_{m,n,h}$
has zero minimal output entropy if and only if the index $h$ is zero.
\item For $m,n,h\in\mathbb{N}$ such $h\leq\min\{m,n\}$, the channel $\Phi_{m,n,h}$
is not entanglement breaking whenever $m>n$. 
\end{itemize}

\section{the minimal output entropy of eposic channels}

\subsection{Background definitions and results }

A quantum system is represented mathematically by a Hilbert space
$H$ which is described by its {\it state} $\varrho$ , a positive
operator in $End(H)$ that has trace one. A pure state is a rank one
state of $H$, such a state can be written in the form $ww^{*}$ where
$w$ is a unit vector in $H$. If $H$ and $K$ are two Hilbert spaces,
then a {\it quantum channel} $\Phi:End(H)\longrightarrow End(K)$
is a completely positive trace preserving map; such a map carries
the states of $H$ into states of $K$ \cite[ch.5]{key-5}. Any quantum
channel $\Phi:End(H)\rightarrow End(K)$ has a Kraus representation
\cite[p.54-p.56]{key-14}, i.e. a set of operators $\{T_{j}\in End(H,K):\,1\leq j\leq n\}$
satisfying 
\[
\overset{n}{\underset{j=1}{\sum}}T_{j}^{*}T_{j}=I_{{\scriptscriptstyle H}}\qquad\qquad and\qquad\Phi(A)=\overset{n}{\underset{j=1}{\sum}}T_{j}AT_{j}^{*}
\]
 If $\,$$\Phi:End(H)\longrightarrow End(K)$ $\,$is a quantum channel
that has Kraus operators $\{T_{j}:\,1\leq j\leq n\}$, then the image
of a pure state $ww^{*}$ under $\Phi$ can be written in the form
$\Phi(ww^{*})=\overset{n}{\underset{{\scriptscriptstyle j=1}}{\sum}}u_{j}u_{j}^{*}$
$\,$where $u_{j}=T_{j}w\in K$.

\begin{notation}
\label{notation 2.1} For a pure state $ww^{*}$ and a quantum channel
$\Phi$, we denote the set$\{u_{j}=T_{j}w\,:1\leq j\leq n\}$ defined
above, by $U_{\Phi,ww^{*}}$.
\end{notation}

\begin{rem}
\label{rem:2.2} For a quantum channel $\Phi$ and a unit vector $w$,
since $\Phi(ww^{*})=\underset{{\scriptscriptstyle j=1}}{\overset{n}{\sum}}u_{j}u_{j}^{*}$
must be a state then $U_{\Phi,ww^{*}}$ must contain a nonzero vector. \end{rem}
\begin{defn}
\label{def:2.3}\cite[Ch.11]{key-10} Let $H,K$ be Hilbert spaces
and $\Phi:End(H)\longrightarrow End(K)$ be a quantum channel. {\it The minimal output entropy}
of $\Phi$, denoted by $S_{min}(\Phi)$ is defined by
\[
S_{min}(\Phi)=\underset{{\scriptstyle w\in H^{1}}}{\min}S(\Phi(ww^{*}))
\]
where $H^{1}$ is the set of all unit vectors in $H$, and where $S(\mathbf{\varrho})=-tr(\mathbf{\varrho\log}\mathbf{\varrho})$
is the von Neumann entropy of the state $\varrho$.
\end{defn}

\begin{rem}
\label{rem:2.4} For a state $\varrho$, the von Neumann entropy $S(\varrho)=\overset{}{\underset{i}{\sum}}-\lambda_{i}\log_{2}\lambda_{i}$
where $\left\{ \lambda_{i}\right\} _{i}$ are the eigenvalues of $\varrho$.
By convention, $0\ln0=0$.
\end{rem}
The following lemma can be proved easily by contradiction.
\begin{lem}
\label{lem:2.5} Let $H$ be a Hilbert space. If $u$ and $v$ are
two linearly independent vectors in $H$, then $uu^{*}$ and $vv^{*}$are
linearly independent.
\end{lem}

\begin{prop}
\label{pro:2.6} Let $H$ and $K$ be Hilbert spaces, and $\Phi:End(H)\longrightarrow End(K)$
be a quantum channel. Then
\begin{enumerate}
\item \textup{$S_{min}(\Phi)=0$} if and only if there exist a pure state
$ww^{*}$ of $H$ such that $\Phi(ww^{*})$ is pure.
\item If for each pure state $ww^{*}$ of $H$, the set $U_{\Phi,ww^{*}}$
contains at least two linearly independent vectors, then $S_{min}(\Phi)\neq0$.
\end{enumerate}
\end{prop}
\begin{proof}
$\,$ 

By continuity of the von Neumann entropy, and compactness of the set
of states \cite[p.29]{key-14}, the minimal output entropy is achieved.
Thus, if $S_{min}(\Phi)=0$, then there is a pure state $ww^{*}$
such that $S(\Phi(ww^{*}))=0$. By \cite[Thm 11.8]{key-10}, $\Phi(ww^{*})$
is a pure state. The other direction follows from the definition of
$S_{min}(\Phi)$. To show the second statement, let $ww^{*}$ be a
pure state. As the set $U_{\Phi,ww^{*}}$ has at least two linearly
independent vectors, by Lemma \ref{lem:2.5}, the state $\Phi(ww^{*})=\overset{}{\underset{{\scriptscriptstyle u_{j}}}{\sum}}u_{j}u_{j}^{*}$
has rank at least two. Hence, $\Phi(ww^{*})$ is not pure for any
pure state $ww^{*}$. The result follows from this and (1).
\end{proof}

\begin{defn}
\cite{key-3,key-6} Let $G$ be a group, and $\pi_{{\scriptscriptstyle H}},\pi_{{\scriptscriptstyle K}}$
be two representations of $G$ on the Hilbert spaces $H$ and $K$.
The quantum channel $\Phi:End(H)\longrightarrow End(K)$ is a $G$-
{\it covariant channel}, if 
\[
\Phi(\pi_{{\scriptscriptstyle H}}{\scriptstyle (g)}A\pi_{{\scriptscriptstyle H}}^{*}{\scriptstyle (g)})=\pi_{{\scriptscriptstyle K}}{\scriptstyle (g)}\Phi(A)\pi_{{\scriptscriptstyle K}}^{*}{\scriptstyle (g)}
\]
for all $A\in End(H)$ and $g\in G$. If both $\pi_{{\scriptscriptstyle H}}$
and $\pi_{{\scriptscriptstyle K}}$ are irreducible, the channel $\Phi$
is called $G$-irreducibly covariant.
\end{defn}

\subsection{EPOSIC channels}

In the following, we give the definition of EPOSIC channel \cite{key-1}.
We begin by reviewing the irreducible representation of $SU(2)$.
For $m\in\mathbb{N}$, let $P_{{\scriptscriptstyle m}}$ denote the
space of homogeneous polynomials of degree $m$ in the two variables
$x_{1},x_{2}$. It is a complex vector space of dimension $m+1$ with
a basis consist of $\left\{ x_{1}^{i}x_{2}^{m-i}:0\leq i\leq m\right\} $,
the space $P_{-1}$ will denote the zero vector space. 

For $m\in\mathbb{N}$, the compact group
\[
SU(2)=\left\{ \tiny\left[{\scriptstyle \begin{array}{cc}
a & b\\
-\bar{b} & \bar{a}
\end{array}}\right]:a,b\in\mathbb{C},\left|a\right|^{2}+\left|b\right|^{2}=1\right\} 
\]
has a representation $\rho_{{\scriptscriptstyle m}}$ on $P_{{\scriptscriptstyle m}}$
given for $f\in P_{{\scriptscriptstyle m}}$ and $g\in SU(2)$ by
\begin{equation}
\left(\rho_{{\scriptscriptstyle m}}{\scriptstyle (g)}f\right){\scriptstyle \left({\scriptstyle x_{1},x_{2}}\right)}=f{\scriptstyle \left(\left({\scriptstyle x_{1},x_{2}}\right){\textstyle g}\right)}=f(ax_{1}-\bar{b}x_{2},bx_{1}+\bar{a}x_{2})
\end{equation}

For each $m\in\mathbb{N}$, $\rho_{{\scriptscriptstyle m}}$ is a
unitary representation with respect to the inner product on $P_{{\scriptscriptstyle m}}$
given by
\begin{equation}
\left\langle x_{1}^{l}x_{2}^{m-l},\,x_{1}^{k}x_{2}^{m-k}\right\rangle _{P_{m}}=l!\,(m-l)!\,\delta_{lk}
\end{equation}
The set $\{\rho_{{\scriptscriptstyle m}}:m\in\mathbb{N}\}$ constitutes
the full list of the irreducible representations of $SU(2)$, see
\cite[ p.276-p.279]{key-13}.

To facilitate the computations, we choose the orthonormal basis of
$P_{{\scriptscriptstyle m}}$ given by the functions
\[
\left\{ f_{{\scriptscriptstyle l}}^{{\scriptscriptstyle m}}=a_{m}^{l}x_{1}^{l}x_{2}^{m-l}:\,0\leq l\leq m\right\} 
\]
with $a_{m}^{l}=\dfrac{{\scriptstyle 1}}{\sqrt{{\scriptstyle l!(m-l)!}}}$.
We call this basis, the standard basis of the $SU(2)$-irreducible
space $P_{{\scriptscriptstyle m}}$. The corresponding standard basis
of $End(P_{{\scriptscriptstyle m}})$ will be 
\[
\{E_{lk}=f_{{\scriptscriptstyle l-1}}^{{\scriptscriptstyle m}}f_{{\scriptscriptstyle k-1}}^{{\scriptscriptstyle m^{*}}}:1\leq l,k\leq m+1\}
\]

For the rest of this paper, we systematically use the following notations
without further mention.

\begin{notation}
\label{Notation 2.8} For $m,n,h,i,j\in\mathbb{N}$ with $0\leq h\leq\min\{m,n\}$,
$0\leq i\leq m+n-2h$, and $0\leq j\leq n$. Let
\begin{itemize}
\item $r=m+n-2h$ ,
\item $B(i):=\{{\scriptstyle j:{\scriptstyle {\scriptstyle \max\{0,-m+i+h\}}}\leq j\leq\min\{i+h,\,n\}}\}$,
\item $l_{ij}:=i-j+h$ ,
\item $\beta_{i,s,j}^{m,n,h}={\scriptstyle (-1)}^{{\scriptstyle {\scriptscriptstyle s}}}\:\sqrt{\tfrac{c_{m,n,h}\,r!\ m!\ n!}{\binom{r}{i}\,\binom{m}{i-j+h}\,\binom{n}{j}}}\:\,\tfrac{\tbinom{h}{s}\,\tbinom{n-h}{j-s}\,\tbinom{m-h}{i-j+s}}{(m-h)!}$
,
\item $\varepsilon_{i}^{j}{\scriptscriptstyle (m,n,h)}\,:=\varepsilon_{i}^{j}=\overset{{\scriptstyle {\scriptscriptstyle \min\{h,j,j+m-i-h\}}}}{\underset{{\scriptscriptstyle {\scriptscriptstyle s=\max\{0,j-i,j+h-n\}}}}{\sum}}\beta_{{\scriptscriptstyle i},s,j}^{m,n,h}$
, and 
\item $\{f_{{\scriptscriptstyle s}}^{k}:0\leq s\leq k\}$ be the standard
basis of $P_{{\scriptscriptstyle k}}$.
\end{itemize}
\end{notation}

$\quad$

\begin{defn}
\label{def:2.9} For $m,n,h\in\mathbb{N}$ with $0\leq h\leq\min\{m,n\}$.
For $0\leq j\leq n$, define the map $T_{j}:P_{r}\longrightarrow P_{{\scriptscriptstyle m}}$
by 
\[
T_{j}(f_{i}^{r})=\left\{ \begin{array}{ccccc}
\varepsilon_{i}^{j}f_{l_{ij}}^{{\scriptscriptstyle m}} &  &  & if & j\in B(i)\\
0 &  &  &  & otherwise
\end{array}\right.
\]

\end{defn}

\begin{prop}
\cite{key-1} The operators $\left\{ T_{j}:0\leq j\leq n\right\} $
in the above definition form Kraus operators for a quantum channel
\[
\Phi_{m,n,h}:End(P_{r})\rightarrow End(P_{{\scriptscriptstyle m}})
\]

\end{prop}
The channel $\Phi_{m,n,h}$ is called EPOSIC channel, and the Kraus
operators given in the above definition are called the EPOSIC Kraus
operators.

\begin{lem}
\label{lem:2.11}\cite{key-1} 
\begin{enumerate}
\item The EPOSIC channel $\Phi_{m,n,h}:End(P_{r})\longrightarrow End(P_{{\scriptscriptstyle m}})$
is an $SU(2)$-irreducibly covariant channel.
\item For each $0\leq i\leq r$, we have
\[
\Phi_{m,n,h}(f_{{\scriptscriptstyle i}}^{{\scriptscriptstyle r}}f_{{\scriptscriptstyle i}}^{{\scriptscriptstyle r^{*}}})=\underset{{\scriptscriptstyle {\scriptscriptstyle j=\max\{0,-m+i+h\}}}}{\overset{{\scriptscriptstyle \min\{i+h,n\}}}{\sum}}(\varepsilon_{i}^{j}{\scriptscriptstyle (m,n,h)})^{2}f_{l_{ij}}^{{\scriptscriptstyle m}}f_{l_{ij}}^{{\scriptscriptstyle m^{*}}}
\]

\end{enumerate}
\end{lem}

\begin{lem}
\label{lem:2.12}\cite[Remark 4.6]{key-1} Let $m,n,h\in\mathbb{N}$
with $0\leq h\leq\min\{m.n\}$ , and $\{T_{j}:0\leq j\leq n\}$ be
the EPOSIC Kraus operators of $\Phi_{m,n,h}$. 
\begin{enumerate}
\item For $0\leq j\leq n$, we have 
\[
T_{j}=\underset{{\scriptscriptstyle i=\max\{0,j-h\}}}{\overset{{\scriptscriptstyle \min\{r,m-h+j\}}}{\sum}}\varepsilon_{i}^{j}f_{{\scriptscriptstyle l_{ij}}}^{{\scriptscriptstyle m}}f_{{\scriptscriptstyle i}}^{r^{*}}.
\]

\item For $w=\overset{r}{\underset{{\scriptscriptstyle i=0}}{\sum}}w_{i}f_{{\scriptscriptstyle i}}^{{\scriptscriptstyle r}}\in P_{{\scriptscriptstyle r}}$,
we have 
\[
T_{j}w=\underset{{\scriptscriptstyle i=\max\{0,j-h\}}}{\overset{{\scriptscriptstyle \min\{r,m-h+j\}}}{\sum}}w_{i}\varepsilon_{i}^{j}f_{{\scriptscriptstyle l_{ij}}}^{{\scriptscriptstyle m}}
\]

\end{enumerate}
\end{lem}
For more details about EPOSIC channel, we refer the reader to \cite{key-1}.

\subsection{The minimal output entropy of $\Phi_{m,n,h}$}

In this section, we determine the EPOSIC channels with zero minimal
output entropy. Namely, we show that the minimal output entropy is
zero if and only if the index $h$ in $\Phi_{m,n,h}$ is zero.
\begin{prop}
\label{pro:2.13} For $m,n\in\mathbb{N}$, the channel $\Phi_{m,n,0}$
has zero minimal output entropy.\end{prop}
\begin{proof}
$\,$

For $k\in\mathbb{N}$, let $\{f_{{\scriptscriptstyle i}}^{k}:{\scriptstyle 0\leq i\leq k}\}$
denote the standard basis for $P_{{\scriptscriptstyle k}}$. By Lemma
\ref{lem:2.11}, we have
\[
\Phi_{m,n,0}(f_{{\scriptscriptstyle 0}}^{{\scriptscriptstyle r}}f_{{\scriptscriptstyle 0}}^{{\scriptscriptstyle r^{*}}})=\underset{{\scriptscriptstyle {\scriptscriptstyle j=\max\{0,-m\}}}}{\overset{{\scriptscriptstyle \min\{0,n\}}}{\sum}}(\varepsilon_{0}^{j}{\scriptscriptstyle (m,n,0)})^{2}f_{l_{ij}}^{{\scriptscriptstyle m}}f_{l_{ij}}^{{\scriptscriptstyle m^{*}}}=f_{{\scriptscriptstyle 0}}^{{\scriptscriptstyle m}}f_{{\scriptscriptstyle 0}}^{{\scriptscriptstyle m^{*}}}
\]
i.e $\Phi_{m,n,0}(f_{{\scriptscriptstyle 0}}^{{\scriptscriptstyle r}}f_{{\scriptscriptstyle 0}}^{{\scriptscriptstyle r^{*}}})$
is a pure state. The result follows by Proposition \ref{pro:2.6}.
\end{proof}

The proof of the following proposition is purely technical calculations,
which we defer to the appendix. 
\begin{prop}
\label{pro:2.14}Let $m,n,h\in\mathbb{N}$ with $0<h\leq\min\{m,n\}$,
and $\Phi_{m,n,h}$ be the associated EPOSIC channel. For any pure
state $ww^{*}\in End(P_{r})$, the set $U_{\Phi_{m,n,h},ww^{*}}$
contains at least two linearly independent vectors.
\end{prop}
By Proposition \ref{pro:2.14}, and Proposition \ref{pro:2.6}, we
have 
\begin{cor}
For strictly positive integers $m,n$ and for $0<h\leq\min\{m,n\}$,
the minimal output entropy $S_{min}(\Phi_{m,n,h})$ is non zero.
\end{cor}

The following theorem summarizes the results of this section.
\begin{thm}
Let $m,n,h\in\mathbb{N}$ with $0\leq h\leq\min\{m,n\}$, and $\Phi_{m,n,h}$
be the associated EPOSIC channel. Then $S_{min}(\Phi_{m,n,h})=0$
if and only if $h=0$.
\end{thm}

\subsection{Computing the minimal output entropy for special cases}

In this section we compute the minimal output entropy of $\Phi_{m,1,1}$
for $m\in\mathbb{N}\smallsetminus\{0\}$. We start by computing the
eigenvalues of $\Phi_{m,1,1}(ww^{*})$ for any pure state $ww^{*}$
then minimizing $S\left(\Phi_{m,1,1}(ww^{*})\right)$ over such states. 

As the channel $\Phi_{m,1,1}:End(P_{{\scriptscriptstyle m-1}})\longrightarrow End(P_{{\scriptscriptstyle m}})$
has only two Kraus operators \cite{key-1}, for any pure state $ww^{*}\in End(P_{{\scriptscriptstyle m-1}})$,
we have
\[
\Phi_{m,1,1}(ww^{*})=u_{{\scriptscriptstyle 0}}u_{{\scriptscriptstyle 0}}^{*}+u_{{\scriptscriptstyle 1}}u_{{\scriptscriptstyle 1}}^{*}
\]
By Proposition \ref{pro:2.14}, the vectors $u_{{\scriptscriptstyle 0}},u_{{\scriptscriptstyle 1}}$
are linearly independent in $P_{{\scriptscriptstyle m}}$. Complete
$u_{0},u_{1}$ to a basis $\{u_{0},u_{1},u_{2},u_{3},...,u_{m}\}$
for $P_{{\scriptscriptstyle m}}$, where $\{u_{2},u_{3},...,u_{m}\}$
is an orthonormal basis for $\{u_{0},u_{1}\}^{\perp}$. Writing the
matrix $\Phi_{m,1,1}(ww^{*})$ in the basis $\{u_{0},u_{1},.....u_{m}\}$
we get the $(m+1)\times(m+1)$ matrix given by

\[
\Lambda_{\Phi_{m,1,1}}:=\left({\scriptstyle \begin{array}{cccccc}
\left\langle u_{0}\left|u_{0}\right.\right\rangle  & \left\langle u_{0}\left|u_{1}\right.\right\rangle  & 0 &  & \cdots & 0\\
\left\langle u_{1}\left|u_{0}\right.\right\rangle  & \left\langle u_{1}\left|u_{1}\right.\right\rangle  & 0 &  & \cdots & 0\\
0 & 0 & 0 &  & \cdots & 0\\
\vdots & \vdots & \vdots &  &  & \vdots\\
0 & 0 & 0 &  & \cdots & 0
\end{array}}\right)
\]
whose nonzero eigenvalues are eigenvalues of 
\[
\left({\scriptstyle \begin{array}{cc}
\left\langle u_{0}\left|u_{0}\right.\right\rangle  & \left\langle u_{0}\left|u_{1}\right.\right\rangle \\
\left\langle u_{1}\left|u_{0}\right.\right\rangle  & \left\langle u_{1}\left|u_{1}\right.\right\rangle 
\end{array}}\right)
\]

By the definition of $\Lambda_{\Phi_{m,1,1}}$, we have 
\begin{lem}
\label{lem:2.17}Let $ww^{*}$ be a pure state in $End(P_{{\scriptscriptstyle m-1}})$.
The non zero eigenvalues of $\Phi_{m,1,1}(ww^{*})$ are\textup{ given
by
\[
\lambda_{1,2}=\frac{1\pm\sqrt{1-4R}}{2}
\]
 where $R=\left\Vert u_{0}\right\Vert ^{2}\left\Vert u_{1}\right\Vert ^{2}-\left|\left\langle u_{0}\left|u_{1}\right.\right\rangle \right|^{2}$.}
\end{lem}

By concavity of von Neumann entropy \cite[ch.11]{key-10} and by \cite[Prop.13.4]{key-14},
the von Neumann entropy of $\Phi_{m,1,1}(ww^{*})$ achieves its minimum
when the difference between $\lambda_{1}$ and $\lambda_{2}$ is maximal,
this is when $R$ takes its minimal value. The following lemma whose
proof was deferred to the appendix, gives the minimal value of $R$.
\begin{lem}
\label{lem:2.18} Let $m\in\mathbb{N}\smallsetminus\{0\}$. For a
pure state $ww^{*}\in End(P_{{\scriptscriptstyle m-1}})$. If $u_{0},u_{1}$
are the elements in $U_{{\scriptscriptstyle \Phi_{m,1,1},ww^{*}}}$\textup{
then} the minimal value of $\left\Vert u_{0}\right\Vert ^{2}\left\Vert u_{1}\right\Vert ^{2}-\left|\left\langle u_{0}\left|u_{1}\right.\right\rangle \right|^{2}$
is $\frac{m}{(m+1)^{2}}$.
\end{lem}

Consequently, the state that minimize von Neumann entropy is the state
with the eigenvalues 
\[
\lambda_{1,2}=\{\frac{{\scriptstyle 1}}{{\scriptstyle m+1}}\,,\,\frac{{\scriptstyle m}}{{\scriptstyle m+1}}\}
\]
By Definition \ref{def:2.3} and Remark \ref{rem:2.4}, we get

\[
S_{min}(\Phi_{m,1,1})=-[\frac{{\scriptstyle 1}}{{\scriptstyle m+1}}\log_{2}\frac{{\scriptstyle 1}}{{\scriptstyle m+1}}+\frac{{\scriptstyle m}}{{\scriptstyle m+1}}\log_{2}\frac{{\scriptstyle m}}{{\scriptstyle m+1}}]
\]

\section{entanglement breaking property of eposic channels}

\subsection{Background definitions and results}

A property of quantum channels that has been studied and used to classify
the quantum channel is their ability to eliminate the entanglement
between the input states of composite systems. Such channels are called
the Entanglement Breaking Trace preserving channels denoted by E.B.T.
Here is a description by \noun{P.}Shor \cite{key-11} for the E.B.T
channels 
\begin{quote}
`` Entanglement breaking channels are channels which destroy entanglement
with other quantum systems. That is, when the input state is entangled
between the input space $H_{in}$ and another quantum system $H_{ref}$
, the output of the channel is no longer entangled with the system
$H_{ref}$ .''
\end{quote}

\begin{lem}
Let $H,K$ be Hilbert spaces and $\Phi:End(H)\longrightarrow End(K)$
be a quantum channel. For $n\in\mathbb{N}$, the map $\Phi\otimes I_{n}:End(H\otimes\mathbb{C}^{n})\longrightarrow End(K\otimes\mathbb{C}^{n})$
defined by taking $A\otimes B$ to $\Phi(A)\otimes B$ and extends
by linearity is a quantum channel.
\end{lem}

\begin{defn}
Let $H_{1}$ and $H_{2}$ be Hilbert spaces. A state $\varrho\in D(H_{1}\otimes H_{2})$
is said to be {\it separable state} if it can be written as a convex
combination of states of the form $\sigma\otimes\tau$ where $\sigma\in D(H_{1})$,
$\tau\in D(H_{2})$. A non-separable state is called an entangled
state. 
\end{defn}

\begin{defn}
\cite{key-8} Let $H,K$ be Hilbert spaces. A quantum channel $\Phi:End(H)\longrightarrow End(K)$
is said to be entanglement breaking if $\Phi\otimes I_{n}(\varrho)$
is separable for any $\varrho\in D(H\otimes\mathbb{C}^{n})$ and $n\in\mathbb{N}$. 
\end{defn}

Recall that in a finite dimensional setting, a characterization of
a quantum channel $\Phi$ is given by its Choi matrix \cite{key-14},
a matrix that is given by
\[
C(\Phi)=\overset{{\scriptscriptstyle d_{H}}}{\underset{{\scriptscriptstyle i,j=1}}{\sum}}\Phi(E_{ij})\otimes E_{ij}
\]
where $E_{ij}$ is the standard basis for $End(H)$. The following
proposition is rephrasing of Theorem 4 in \cite{key-8}.
\begin{prop}
\label{pro:3.4} Let $H,K$ be Hilbert spaces and $\Phi:End(H)\longrightarrow End(K)$
is a quantum channel. The following statements are equivalent
\begin{enumerate}
\item $\Phi$ is an E.B.T channel.
\item The Choi matrix of $\Phi$ is separable. 
\item $\Phi$ can be written in operator sum form using only Kraus operators
of rank one.
\end{enumerate}
\end{prop}

By \cite[Prop 5.2  and Thm5.3]{key-14}, we have 
\begin{lem}
\label{lem:3.5} Let $H$ and $K$ be Hilbert spaces, and $\Phi:End(H)\longrightarrow End(K)$
be a quantum channel. The rank of the Choi matrix of $\Phi$ is an
achievable lower bound for the number of Kraus operators of $\Phi$. 
\end{lem}

The following proposition follows directly by \cite[Thm 1]{key-9}
and Proposition \ref{pro:3.4}. The corollary to it, is just a generalization
of \cite[Theroem 6]{key-8}.
\begin{prop}
\label{pro.3.6} Let $H$,$K$ be Hilbert spaces of dimension $d_{{\scriptscriptstyle H}},d_{{\scriptscriptstyle K}}$,
and $\Phi:End(H)\rightarrow End(K)$ be a quantum channel. If $rankC(\Phi)<\max\{d_{{\scriptscriptstyle H}},rank(Tr_{\overline{H}}(C(\Phi)))\}$
then $\Phi$ is not E.B.T
\end{prop}

By Lemma \ref{lem:3.5}, and Proposition \ref{pro.3.6}, we get
\begin{cor}
\label{cor:3.7} Let $H,K$ be Hilbert spaces of dimension $d_{{\scriptscriptstyle H}}$,$d_{{\scriptscriptstyle K}}$
respectively such that $d_{{\scriptscriptstyle H}}\geq d_{{\scriptscriptstyle K}}$.
Let $\Phi:End(H)\rightarrow End(K)$ be a quantum channel. If $\Phi$
can be written in Kraus operator fewer than $d_{{\scriptscriptstyle H}}$
then $\Phi$ is not E.B.T
\end{cor}

Let $H,K$ be Hilbert spaces, let $\Phi^{*}$ denote the dual map
of the quantum channel $\Phi:End(H)\longrightarrow End(K)$. It is
evident that if $\{T_{j}:1\leq j\leq k\}$ is Kraus operators for
$\Phi$ then $\{T_{j}^{*}:1\leq j\leq k\}$ will be Kraus operators
for $\Phi^{*}$. As
\[
T_{j}=uv^{*}\Longleftrightarrow T_{j}^{*}=vu^{*}
\]
then by Proposition \ref{pro:3.4}(3), we have 
\begin{lem}
\label{lem:3.8} Let $\Phi$ be a quantum channel then $\Phi$ is
an E.B.T map if and only if its dual $\Phi^{*}$ is an E.B.T map.
\end{lem}

\subsection{The E.B.T property of EPOSIC channels.}

In this section, we classify EPOSIC channels according to their E.B.T
property. We didn't obtain a full classification, we state below the
partial results that we obtained.

\begin{thm}
For $m\in\mathbb{N}$, the channel $\Phi_{m,m,m}$ and $\Phi_{0,m,0}$
are E.B.T channels.\end{thm}
\begin{proof}
$\,$ 

Let $\{T_{j},0\leq j\leq m\}$ be EPOSIC Kraus operators for $\Phi_{m,m,m}$.
By Lemma \ref{lem:2.12}, we have 
\[
rank(T_{j})\leq\min\left\{ 0,\,m,\,j,\,m-j\right\} +1=1
\]
for any $0\leq j\leq m$. So, by Proposition \ref{pro:3.4}(3), the
channel $\Phi_{m,m,m}$ is E.B.T. As by \cite[Sec.5]{key-1} we have
\[
\Phi_{m,m,m}^{*}=\frac{1}{m+1}\Phi_{0,m,0}
\]
 then by Lemma \ref{lem:3.8} the channel $\Phi_{0,m,0}$ is also
E.B.T.
\end{proof}

\begin{prop}
\label{pro:3.10} Let $m,n,h\in\mathbb{N}$ with $0\leq h\leq\min\{m,n\}$
\begin{enumerate}
\item If $n\geq2h$ then $\Phi_{m,n,h}$ is not E.B.T for any $m>2h$.
\item If $n\leq2h$ then $\Phi_{m,n,h}$ is not E.B.T for any $m>n$.
\end{enumerate}
\end{prop}
\begin{proof}
$\,$ 

The channel
\[
\Phi_{m,n,h}:End(P_{{\scriptscriptstyle r}})\longrightarrow End(P_{{\scriptscriptstyle m}})
\]
has $n+1$ EPOSIC Kraus operators. If $n\geq2h$ then $dim(P_{{\scriptscriptstyle r}})\geq dim(P_{{\scriptscriptstyle m}})$,
and by Corollary \ref{cor:3.7}, we get that $\Phi_{m,n,h}$ is not
E.B.T whenever $m>2h$. If $n\leq2h$ then $r\leq m$ and by (1) the
channel 
\[
\Phi_{r,n,n-h}:End(P_{m})\longrightarrow End(P_{r})
\]
 is not E.B.T whenever $r>2(n-h)$ i.e whenever $m>n$. As by \cite[Sec.5]{key-1}
we have 
\[
\Phi_{m,n,h}=\Phi_{r,n,n-h}^{*}
\]
then by Lemma \ref{lem:3.8} the channel $\Phi_{m,n,h}$ is not E.B.T
whenever $m>n$.
\end{proof}

\begin{cor}
Let $m,n,h\in\mathbb{N}$ with $0\leq h\leq\min\{m,n\}$. The channel
$\Phi_{m,n,h}$ is not E.B.T whenever $m>n$. In Particular, $\Phi_{m,h,h}$
is not E.B.T for any $0\leq h<m$.
\end{cor}

\section*{acknowledgment}

We thank professors B. Collins and T. Giordano for their help and
advice. We also would like to acknowledge the financial support from
king Faisal university.

\appendix

\section{$\,$}

\subsection{Proof of Proposition \ref{pro:2.14}.}

For the proof of Proposition \ref{pro:2.14}, the following lemmas
are needed. The first one can be proved by direct computation using
the formula
\[
\varepsilon_{i}^{j}{\scriptscriptstyle (m,n,h)}=\overset{{\scriptstyle {\scriptscriptstyle \min\{h,j,j+m-i-h\}}}}{\underset{{\scriptscriptstyle {\scriptscriptstyle s=\max\{0,j-i,j+h-n\}}}}{\sum}}\beta_{{\scriptscriptstyle i},s,j}^{m,n,h}
\]

\begin{lem}
\label{lem:4.1} For $m,n,h\in\mathbb{N}$ with $0\leq h\leq\min\{m,n\}$,
and $r=m+n-2h$ , then 
\begin{enumerate}
\item $\varepsilon_{i}^{{\scriptscriptstyle 0}}\neq0$, $\qquad\;\:$$\quad$for
$\,$ $0\leq i\leq m-h$.
\item $\varepsilon_{i}^{{\scriptscriptstyle i-m+h}}\neq0$, $\quad$for$\,$
$m-h\leq i\leq r$.
\item $\varepsilon_{i}^{{\scriptscriptstyle i+h}}\neq0$, $\qquad\;\:$for$\,$
$0\leq i\leq n-h$. 
\item $\varepsilon_{i}^{n}\neq0$, $\qquad\;\:$$\;$for$\,$ $n-h\leq i\leq r$.
\end{enumerate}
\end{lem}

\begin{lem}
\label{lem:4.2} Let $m,n,h\in\mathbb{N}$ with $0<h\leq\min\{m,n\}$
and $r=m+n-2h$ . For any $0\leq i\leq r$, we have 
\[
{\textstyle \max\{0,-m+i+h\}}<\min\{i+h,\,n\}
\]
\end{lem}
\begin{proof}
$\,$

Let
\[
j_{1}=\max\{0,-m+i+h\}\qquad\quad\quad\quad
\]
\[
\qquad\quad=\left\{ \begin{array}{cccccc}
0 &  & if &  & 0\leq i\leq m-h\\
i-m+h &  & if &  & m-h\leq i\leq r
\end{array}\right.
\]
and
\[
j_{2}=\min\{i+h,\,n\}=\left\{ \begin{array}{cccccc}
i+h &  & if &  & 0\leq i\leq n-h\\
n &  & if &  & n-h\leq i\leq r
\end{array}\right.
\]
If $j_{1}=0$, then $j_{1}<h\leq j_{2}$. Otherwise, 
\[
j_{1}=i-(m-h)\leq\min\{i,\,r-m+h\}\quad
\]
\[
\qquad\qquad\;=\min\{i,\,n-h\}
\]
\[
\qquad\qquad\qquad<\min\{i+h,\,n\}=j_{2}
\]

\end{proof}

Recall the definition of
\[
B(i)=\{j:{\textstyle \max\{0,-m+i+h\}}\leq j\leq\min\{i+h,\,n\}\}
\]

\begin{cor}
\label{cor:4.3} For $m,n,h\in\mathbb{N}$ with $0<h\leq\min\{m,n\}$,
let $r=m+n-2h$. For each $0\leq i\leq r$, there exist $j_{1},j_{2}\in B(i)$
such that $j_{1}<j_{2}$, and 
\[
\varepsilon_{i}^{j_{1}}\neq0,\qquad\varepsilon_{i}^{j_{2}}\neq0
\]
\end{cor}
\begin{proof}
$\,$

Let $j_{1}=\max\{0,-m+i+h\}$ and $j_{2}=\min\{i+h,\,n\}.$ Both $j_{1},j_{2}\in B(i)$,
and by Lemma \ref{lem:4.2} we have $j_{1}<j_{2}$. Lemma \ref{lem:4.1}
gives that both $\varepsilon_{i}^{j_{1}}$ and $\varepsilon_{i}^{j_{2}}$
are nonzero.
\end{proof}

Next we give the proof of Proposition \ref{pro:2.14}.
\begin{prop}
Let $m,n,h\in\mathbb{N}$ with $0<h\leq\min\{m,n\}$, and $\Phi_{m,n,h}$
be the associated EPOSIC channel. For any pure state $ww^{*}\in End(P_{r})$,
the set $U_{\Phi_{m,n,h},ww^{*}}$ contains at least two linearly
independent vectors.\end{prop}
\begin{proof}
$\,$

Let $r=m+n-2h$ and $ww^{*}$ be any pure state in $End(P_{{\scriptscriptstyle r}})$
for some unit vector $w\in P_{r}$. Let $w=\overset{{\scriptscriptstyle r}}{\underset{{\scriptscriptstyle i=0}}{\sum}}w_{i}f_{{\scriptscriptstyle i}}^{{\scriptscriptstyle r}}$
where $\overset{{\scriptscriptstyle r}}{\underset{{\scriptscriptstyle i=0}}{\sum}}\left|w_{i}\right|^{2}=1$,
and $i_{1}$ be the smallest index $i$ such that $w_{i}\neq0$. By
Corollary \ref{cor:4.3}, there exist $j_{1}<j_{2}\in B(i_{1})$ such
that $\varepsilon_{i_{1}}^{j_{1}}\neq0$ and $\varepsilon_{i_{1}}^{j_{2}}\neq0$.\\
Since $j\in B(i_{1})$ if and only if 
\[
\max\{0,j-h\}\leq i_{1}\leq\min\{r,m-h+j\}
\]
 then by Lemma \ref{lem:2.12}, we have 
\[
u_{j_{1}}=T_{j_{1}}w=\underset{{\scriptscriptstyle i=\max\{0,j_{1}-h\}}}{\overset{{\scriptscriptstyle \min\{r,m-h+j_{1}\}}}{\sum}}w_{i}\varepsilon_{i}^{j_{1}}f_{{\scriptscriptstyle i-j_{1}+h}}^{{\scriptscriptstyle m}}\neq0
\]
and
\[
u_{j_{2}}=T_{j_{2}}w=\underset{{\scriptscriptstyle i=\max\{0,j_{2}-h\}}}{\overset{{\scriptscriptstyle \min\{r,m-h+j_{2}\}}}{\sum}}w_{i}\varepsilon_{i}^{j_{2}}f_{{\scriptscriptstyle i-j_{2}+h}}^{{\scriptscriptstyle m}}\neq0
\]
 If $U_{\Phi_{m,n,h},ww^{*}}=\{u_{j}:0\leq j\leq n\}$ does not contain
two linearly independent vectors, then there exist $\alpha\neq0$
such that 
\[
u_{j_{2}}=\alpha u_{j_{1}}
\]
In particular, comparing the coefficients of $f_{{\scriptscriptstyle i_{1}-j_{2}+h}}^{{\scriptscriptstyle m}}$,
we obtain
\[
0\neq w_{i_{1}}\varepsilon_{i_{1}}^{j_{2}}=\alpha w_{i_{2}}\varepsilon_{i_{2}}^{j_{1}}
\]
for some $i_{2}$, where $i_{1}-j_{2}+h=i_{2}-j_{1}+h$. \\ i.e.
$i_{2}=i_{1}-(j_{2}-j_{1})<i_{1}$ and $w_{i_{2}}\neq0$, contradicting
the minimality of $i_{1}$. 
\end{proof}

\subsection{Proof of Lemma \ref{lem:2.18}.}

Some elementary computational lemmas are needed, the following one
follows by direct computations. Item (3) follows from the fact that
$\Phi_{{\scriptscriptstyle m,1,1}}$ is trace preserving.
\begin{lem}
\label{lem:4.5}Let $m\in\mathbb{N}\smallsetminus\{0\}$ then 
\begin{enumerate}
\item For $0\leq l\leq m-1$, we have
\[
{\scriptstyle {\textstyle \varepsilon_{{\scriptscriptstyle l}}^{{\scriptscriptstyle 0}}=\sqrt{\frac{l+1}{m+1}},\qquad\varepsilon_{{\scriptscriptstyle l}}^{{\scriptscriptstyle 1}}=-\sqrt{\frac{m-l}{m+1}}}}
\]
,and 
\[
(\varepsilon_{{\scriptscriptstyle l}}^{{\scriptscriptstyle 0}})^{2}+(\varepsilon_{{\scriptscriptstyle l}}^{{\scriptscriptstyle 1}})^{2}=1
\]

\item For $1\leq l\leq m-1$,
\[
{\textstyle (\varepsilon_{{\scriptscriptstyle l}}^{{\scriptscriptstyle 0}})^{2}=(\varepsilon_{{\scriptscriptstyle l-1}}^{{\scriptscriptstyle 0}})^{2}+\frac{1}{m+1}}
\]
,and 
\[
{\textstyle (\varepsilon_{{\scriptscriptstyle l-1}}^{{\scriptscriptstyle 1}})^{2}=(\varepsilon_{{\scriptscriptstyle l}}^{{\scriptscriptstyle 1}})^{2}+\frac{1}{m+1}}
\]
 
\item \textup{$\left\Vert u_{0}\right\Vert ^{2}+\left\Vert u_{1}\right\Vert ^{2}=1$.}
\end{enumerate}
\end{lem}

\begin{rem}
\label{Re:4.6} By Lemma \ref{lem:2.12} (2), the vectors $u_{0},u_{1}$
for $\Phi_{m,1,1}$ are given by
\[
u_{0}=\overset{{\scriptscriptstyle m}}{\underset{{\scriptscriptstyle l=1}}{\sum}}\varepsilon_{{\scriptscriptstyle l-1}}^{{\scriptscriptstyle 0}}w_{{\scriptscriptstyle l-1}}f_{l}^{{\scriptscriptstyle m}}\qquad and\quad u_{1}=\overset{{\scriptscriptstyle m-1}}{\underset{{\scriptscriptstyle l=0}}{\sum}}\varepsilon_{{\scriptscriptstyle l}}^{{\scriptscriptstyle 1}}w_{{\scriptscriptstyle l}}f_{l}^{{\scriptscriptstyle m}}
\]
\end{rem}
\begin{lem}
\label{lem:4.7} Let $m\in\mathbb{N}\smallsetminus\{0\}$. For a pure
state $ww^{*}\in End(P_{{\scriptscriptstyle m-1}})$, let $u_{0},u_{1}$
be the elements in $U_{{\scriptscriptstyle \Phi_{m,1,1},ww^{*}}}$
\textup{, and $R=\left\Vert u_{0}\right\Vert ^{2}\left\Vert u_{1}\right\Vert ^{2}-\left|\left\langle u_{0}\left|u_{1}\right.\right\rangle \right|^{2}$}.
The minimal value of $R$
\[
\frac{m}{(m+1)^{2}}
\]
\end{lem}
\begin{proof}
$\,$

By Remark \ref{Re:4.6}, we have 
\[
u_{0}=\overset{{\scriptscriptstyle m}}{\underset{{\scriptscriptstyle l=1}}{\sum}}\varepsilon_{{\scriptscriptstyle l-1}}^{{\scriptscriptstyle 0}}w_{{\scriptscriptstyle l-1}}f_{l}^{{\scriptscriptstyle m}}\;and\;u_{1}=\overset{{\scriptscriptstyle m-1}}{\underset{{\scriptscriptstyle l=0}}{\sum}}\varepsilon_{{\scriptscriptstyle l}}^{{\scriptscriptstyle 1}}w_{{\scriptscriptstyle l}}f_{l}^{{\scriptscriptstyle m}}
\]
So 

\[
{\textstyle \left\langle u_{0}\left|u_{1}\right.\right\rangle =\overset{{\scriptscriptstyle m-1}}{\underset{{\scriptscriptstyle l=1}}{\sum}}\varepsilon_{{\scriptscriptstyle l-1}}^{{\scriptscriptstyle 0}}\overline{w}_{{\scriptscriptstyle l-1}}\varepsilon_{{\scriptscriptstyle l}}^{{\scriptscriptstyle 1}}w_{{\scriptscriptstyle l}}=\overset{{\scriptscriptstyle m-1}}{\underset{{\scriptscriptstyle l=1}}{\sum}}\varepsilon_{{\scriptscriptstyle l-1}}^{{\scriptscriptstyle 0}}w_{{\scriptscriptstyle l}}\varepsilon_{{\scriptscriptstyle l}}^{{\scriptscriptstyle 1}}\overline{w}_{{\scriptscriptstyle l-1}}=\left\langle v_{0}\left|v_{1}\right.\right\rangle }
\]
 where
\[
v_{0}=\overset{{\scriptscriptstyle m-1}}{\underset{{\scriptscriptstyle l=1}}{\sum}}\varepsilon_{{\scriptscriptstyle l-1}}^{{\scriptscriptstyle 0}}w_{{\scriptscriptstyle l}}f_{l}^{{\scriptscriptstyle m}}\quad and\quad v_{1}=\overset{{\scriptscriptstyle m-1}}{\underset{{\scriptscriptstyle l=1}}{\sum}}\varepsilon_{{\scriptscriptstyle l}}^{{\scriptscriptstyle 1}}w_{{\scriptscriptstyle l-1}}f_{l}^{{\scriptscriptstyle m}}
\]

As ${\textstyle \left\Vert w\right\Vert }^{2}=\overset{{\scriptscriptstyle m-1}}{\underset{{\scriptscriptstyle l=0}}{\sum}}\left|w_{{\scriptscriptstyle l}}\right|^{2}=1$,
Using Lemma \ref{lem:4.5}, we obtain \\${\displaystyle \left\Vert v_{0}\right\Vert ^{2}=\overset{{\scriptscriptstyle m-1}}{\underset{{\scriptscriptstyle l=1}}{\sum}}\left(\varepsilon_{{\scriptscriptstyle l-1}}^{{\scriptscriptstyle 0}}\right)^{2}\left|w_{{\scriptscriptstyle l}}\right|^{2}}$
\[
=\overset{{\scriptscriptstyle m-1}}{\underset{{\scriptscriptstyle l=1}}{\sum}}\left(\varepsilon_{{\scriptscriptstyle l-1}}^{{\scriptscriptstyle 0}}\right)^{2}\left|w_{{\scriptscriptstyle l}}\right|^{2}+{\textstyle \frac{{\textstyle \left\Vert w\right\Vert }^{2}}{m+1}-\frac{{\textstyle \left\Vert w\right\Vert ^{2}}}{m+1}}\qquad\qquad
\]
\[
\quad\quad={\textstyle \frac{1}{m+1}\left|w_{{\scriptscriptstyle 0}}\right|^{2}+\overset{{\scriptscriptstyle m-1}}{\underset{{\scriptscriptstyle l=1}}{\sum}}\left(\left(\varepsilon_{{\scriptscriptstyle l-1}}^{{\scriptscriptstyle 0}}\right)^{2}+\frac{1}{m+1}\right)\left|w_{{\scriptscriptstyle l}}\right|^{2}-\frac{{\textstyle \left\Vert w\right\Vert }^{2}}{m+1}}
\]
\[
=\left(\varepsilon_{{\scriptscriptstyle 0}}^{{\scriptscriptstyle 0}}\right)^{2}\left|w_{{\scriptscriptstyle 0}}\right|^{2}+\overset{{\scriptscriptstyle m-1}}{\underset{{\scriptscriptstyle l=1}}{\sum}}{\textstyle \left(\varepsilon_{{\scriptscriptstyle l}}^{{\scriptscriptstyle 0}}\right)^{2}\left|w_{{\scriptscriptstyle l}}\right|^{2}-\frac{1}{m+1}}.\qquad\quad
\]
 Thus \\ ${\displaystyle \left\Vert v_{0}\right\Vert ^{2}=\overset{{\scriptscriptstyle m-1}}{\underset{{\scriptscriptstyle l=0}}{\sum}}\left(\varepsilon_{{\scriptscriptstyle l}}^{{\scriptscriptstyle 0}}\right)^{2}\left|w_{{\scriptscriptstyle l}}\right|^{2}-{\textstyle \frac{1}{m+1}}}$
\[
=\overset{{\scriptscriptstyle m}}{\underset{{\scriptscriptstyle l=1}}{\sum}}\left(\varepsilon_{{\scriptscriptstyle l-1}}^{{\scriptscriptstyle 0}}\right)^{2}\left|w_{{\scriptscriptstyle l-1}}\right|^{2}-\frac{1}{m+1}\qquad\qquad\qquad
\]
$\qquad\quad=\left\Vert u_{0}\right\Vert ^{2}-\frac{1}{m+1}$

Similarly$\qquad$ $\left\Vert v_{1}\right\Vert ^{2}=\left\Vert u_{1}\right\Vert ^{2}-\frac{1}{m+1}$.
So, \\

$\left|\left\langle u_{0}\left|u_{1}\right.\right\rangle \right|^{2}=\left|\left\langle v_{0}\left|v_{1}\right.\right\rangle \right|^{2}\leq\left\Vert v_{0}\right\Vert ^{2}\left\Vert v_{1}\right\Vert ^{2}$
\[
\qquad\qquad\qquad\qquad=\left\Vert u_{0}\right\Vert ^{2}\left\Vert u_{1}\right\Vert ^{2}-{\textstyle \frac{m}{\left(m+1\right)^{2}}}
\]

Thus
\[
R=\left\Vert u_{0}\right\Vert ^{2}\left\Vert u_{1}\right\Vert ^{2}-\left|\left\langle u_{0}\left|u_{1}\right.\right\rangle \right|^{2}\geq\frac{m}{\left(m+1\right)^{2}}
\]

and ${\displaystyle \frac{m}{\left(m+1\right)^{2}}}$ is a lower bound
for $R$. 

$\,$

For the minimal value of $R$, let $w=(1,0,\ldots,0)^{t}$ to get
\[
u_{0}=\sqrt{\frac{1}{m+1}}f_{{\scriptscriptstyle 0}}^{{\scriptscriptstyle m}},\quad u_{1}=\sqrt{\frac{m}{m+1}}f_{{\scriptscriptstyle 1}}^{{\scriptscriptstyle m}}
\]
 and $R=\frac{m}{\left(m+1\right)^{2}}$.
\end{proof}

\end{document}